\documentclass[10pt]{article}

\usepackage{epsfig}
\usepackage{graphicx}
\usepackage{latexsym}
\usepackage{a4wide}
\usepackage{amsthm}
\usepackage{amsmath}
\usepackage{amssymb}

\usepackage{color}

\newtheorem{lemma}{Lemma}
\newtheorem{cor}{Corollary}
\newtheorem{obs}{Observation}
\newtheorem{theorem}{Theorem}
\theoremstyle{definition}
\newtheorem{defi}{Definition}
\renewenvironment{proof}{\noindent\mbox{\newline\smallskip}\noindent\textbf{Proof:}}{$\hfill\square$\newline}

\graphicspath{{./}{../figures/}}

\newcommand{\optSQ}{\ensuremath{B^*}}

\newcommand{\sq}{\textsc{$(p,k)$-Square Covering}}
\newcommand{\rec}{\textsc{$(p,k)$-Rectangle Covering}}
\newcommand{\true}{TRUE }
\newcommand{\false}{FALSE }

\begin{document}


\title{Covering Points by Disjoint Boxes with Outliers
\thanks{This work was supported by Basic Science Research Program through
the National Research Foundation of Korea (NRF) funded by the
Ministry of Education, Science and Technology (No. 2009-0067195)
and by the Brain Korea 21 Project in 2010.
}
}
\author{Hee-Kap Ahn\footnotemark[1]
\and Sang Won Bae\footnotemark[4]
\and Erik D. Demaine\footnotemark[2]
\and Martin L. Demaine\footnotemark[2]
\and Sang-Sub Kim\footnotemark[1]
\and Matias Korman\footnotemark[3]
\and Iris Reinbacher\footnotemark[1]
\and Wanbin Son\footnotemark[1]
}

\date{}
\footnotetext[1]{Department of Computer Science and Engineering, POSTECH, South Korea. \tt{\{heekap, helmet1981, irisrein, mnbiny\}@postech.ac.kr}}
\footnotetext[2]{ MIT Computer Science and Artificial Intelligence Laboratory. \texttt{\{edemaine, mdemaine\}@mit.edu}}
\footnotetext[3]{Computer Science department, Universit\'e Libre de Bruxelles (ULB), Belgium. \texttt{mkormanc@ulb.ac.be}}
\footnotetext[4]{Department of Computer Science, Kyonggi University, Suwon, Korea. \texttt{swbae@kgu.ac.kr}}

\maketitle


\begin{abstract}
  	For a set of $n$ points in the plane, we consider the
  	axis--aligned \textsc{$(p,k)$-Box Covering} problem: Find $p$ axis-aligned, pairwise-disjoint boxes that together contain at least $n-k$ points. In this paper, we consider the boxes to be either squares or rectangles, and we want
  	to minimize the area of the largest box. For general $p$ we show that the problem is NP-hard for both squares and rectangles. For a small, fixed number $p$, we give algorithms that find the solution in the following running times:
  	For squares we have $O(n+k\log k)$ time for $p=1$, and
  	$O(n\log n+k^p\log^p k)$ time for $p = 2,3$. For rectangles we get $O(n + k^3)$ for $p = 1$ and $O(n\log n+k^{2+p}\log^{p-1} k)$ time for $p = 2,3$.
  	In all cases, our algorithms use $O(n)$ space.
\end{abstract}



\section{Introduction}
    Motivated by clustering, we consider the problem of splitting a large set of points into a small number of groups. From a geometric point of view, we want to group points together that are `close' with respect to some distance measure. It is easy to see that the choice of distance measure directly influences the shape of the clusters. Depending on the application, it may be useful to consider only disjoint clusters.  It is important to take noise into account, especially when dealing with raw data. That means, we may want to remove outliers that are `far' from the clusters, or that would unduly influence their shape.

    In this paper, we consider the following optimization problem:
    Given a set $P$ of $n$ points in the plane and two
    integers $p > 0$ and $k\geq 0$, find $p$ pairwise-disjoint squares or
    rectangles that together contain at least $n-k$ points of $P$ and
    minimize the largest area among the $p$ squares or rectangles. We treat the squares or rectangles as closed sets, and although we want them to be pairwise-disjoint, we allow overlap at their boundaries or corners.

    We call this problem the \textsc{$(p,k)$-Square Covering}
    and the \textsc{$(p,k)$-Rectangle Covering} problem, respectively, according to the
    shape of the covering regions.
    The $k$ points that are not covered by a solution of the problem
    are called \emph{outliers}.

    Both problems are variations and/or extensions of the
    \emph{rectilinear $p$-center problem}. This is usually considered as
    the problem of finding $p$ congruent squares of smallest possible size that
    together contain all points of $P$, where the $p$ squares may overlap.
    In our setting, however, we have (1) that the $p$ regions must not
    overlap each other (except at their boundaries) and (2) that up to a predefined number of $k$ points are considered as outliers and can be ignored.
    It is known that the rectilinear $p$-center problem is
    NP-hard even to approximate within ratio
    $1.5$~\cite{ms-cscgl-84}.
    However, for $p\leq 4$, worst-case optimal-time algorithms are
    known: linear time for $p \leq 3$ and $O(n\log n)$ time for
    $p=4$. For $p\geq 5$, the best known time bound is $O(n^{p-4}
    \log^5 n)$~\cite{sw-rpppp-96}.

    For the \textsc{$(p,0)$-Rectangle Covering} problem, less work has
    been done. Bespamyatnikh and Segal~\cite{bs-cspta-00} presented a
    deterministic $O(n \log n)$ time algorithm for $p=2$, but no efficient
    algorithm for $p\geq 3$ is known. Several
    papers considered variations of the \textsc{$(2,0)$-Rectangle
    Covering} problem --- e.g., arbitrary orientation and three or higher
    dimensions --- and achieved efficient algorithms;
    see for example~\cite{ab-cptdr-08,dgn-skper-05,jk-oicps-96,kks-dr2cp-00,sd-csppu-07}.

    Outliers can also be seen as violation of constraints:
    basically, the points in $P$ are constraints to be covered by
    squares or rectangles in our problems and $k$ of them are
    allowed to be violated. In this sense, there is a connection to geometric
    optimization with violated constraints which has been studied by
    several researchers. Matou\v{s}ek~\cite{m-gofvc-95} and
    Chan~\cite{c-ldlpv-05} presented efficient algorithms for
    \emph{LP-type problems} allowing $k$ violated constraints.
    The class of LP-type problems, which extends linear programming in a
    combinatorial sense, was introduced by Sharir and
    Welzl~\cite{sw-cblpr-92}. Also, a deterministic linear-time algorithm
    for LP-type problems of finite LP-dimension is known~\cite{cm-ltdao-96}.
    The LP-dimension is a parameter associated with an LP-type problem;
    for instance, the \textsc{$(1,0)$-Square Covering}
    problem, or equivalently the rectilinear $1$-center problem, has
    LP-dimension $3$ since the smallest unique enclosing square is
    determined by three points of the given point set.
    Indeed, the rectilinear $p$-center problem for $p \leq 3$ is
    known to be an LP-type problem~\cite{sw-rpppp-96}, so
    linear-time algorithms follow. Thus, the \textsc{$(1,k)$-Square Covering}
    problem can be solved in $O(n \log n + k^2 \log^2 n)$ time and the
    \textsc{$(1,k)$-Rectangle Covering} problem in $O(n \log n +
    k^{\frac{11}{4}} n^{\frac{1}{4}}\log^{O(1)} n)$ time, according
    to Chan~\cite{c-ldlpv-05}. For LP-dimension larger than four, no
    efficient algorithm has been found as to date.
    More details on LP-type problems can be found in Sharir and
    Welzl~\cite{sw-cblpr-92}, Matou\v{s}ek and
    \v{S}kovro\v{n}~\cite{ms-tvlptop-03}, and Dyer et~al.~\cite{dm-lpld-04}.

    Independent of LP-type problems with violated constraints,
    there are some previous results dealing with outliers when $p=1$. Aggarwal
    et~al.~\cite{aiks-fkpmd-91} achieved a running time of
    $O((n-k)^2 n\log n)$ using $O((n-k)n)$ space for both the
    \textsc{$(1,k)$-Square Covering} and the \textsc{$(1,k)$-Rectangle Covering}
    problems. Later, Segal and Kedem~\cite{sk-ekpsa-98} gave an
    $O(n+k^2(n-k))$ time algorithm for the \textsc{$(1,k)$-Rectangle Covering}
    problem using $O(n)$ space. A randomized algorithm that runs in $O(n\log n)$
    time was given for the \textsc{$(1,k)$-Square
    Covering} problem by Chan~\cite{c-garot-99}. Most recently,
    Atanassov et al.~\cite{abcmmpsw-aoor-09} presented an $O(n+k^3)$
    time algorithm for the \textsc{$(1,k)$-Rectangle Covering}
    problem.

    Most of the above algorithms are optimal when the number of outliers
    is either a small constant or close to $n$. In this paper, we
    are interested in algorithms with small running time in $k$.
    Ideally, we would also like to preserve optimality in $n$ for small $k$. We summarize the new results shown in this paper:
\begin{itemize}
\item 	\textit{NP-hardness}: In Section~\ref{sec:nphard}, we prove that both
     the \textsc{$(p,k)$-Square Covering} and the
     \textsc{$(p,k)$-Rectangle Covering} problems are NP-hard when
     $p$ is part of the input, even for a fixed $k\geq 0$.
     These are the first NP-hardness proofs for a variant of the rectilinear
     $p$-center problem where the covering regions are disjoint and
     also for the problem of covering points by $p$ rectangles.

\item	\textit{Efficient algorithms for small $p$}: In Section~\ref{sec:smallp}, we give 	efficient algorithms if the number of boxes is small. All our algorithms use linear space. The running times of our algorithms are summarized in Table~\ref{tab_res}. Recall that the previously best known results for this problem with outliers were restricted to only one box: $O(n\log n)$ for the \textsc{$(1,k)$-Square Covering} problem~\cite{c-garot-99}, and $O(n+k^3)$ for the \textsc{$(1,k)$-Rectangle Covering} problem~\cite{abcmmpsw-aoor-09}.

\begin{table*}[h]
\centering \caption{Running times of our \textsc{$(p,k)$-Square/Rectangle
Covering} algorithms} \label{tab_res}
\begin{tabular}{|c|c|c|} \hline
     & Squares & Rectangles \\ \hline
    $p=1$ &$O(n+k\log k)$ &$O(n+k^3)$ \\ \hline
    $p=2$ &$O(n\log n+k^2\log^2 k)$ &$O(n\log n+k^4\log k)$ \\ \hline
    $p=3$ &$O(n\log n +k^3\log^3k)$ &$O(n\log n+k^5\log k)$ \\ \hline
\end{tabular}
\end{table*}

\end{itemize}

\section{A lower bound}\label{sec:prelim}
	We consider the \textsc{$(p,k)$-Square Covering} and the \textsc{$(p,k)$-Rectangle Covering} problem. Given a set $P$ of $n$ points in the plane, and two integers $k \geq 0$ and $p >0$, find $p$ axis--aligned pairwise--disjoint (overlap of boundaries is allowed), closed squares or rectangles, that together cover at least $n-k$ points of $P$, such that the area of the largest square or rectangle is minimized. We refer to the $k$ points that are not contained in the union of all squares or rectangles as \emph{outliers}.

	The algorithms we present in Section~\ref{sec:algs} are efficient, as we can show the following lower bound that holds for both the \textsc{$(p,k)$-Square Covering} and the \textsc{$(p,k)$-Rectangle Covering} problem.

\begin{lemma}\label{lem:lower}
	Let $k \in \mathbb{N}$ be part of the input and let $p$ be any fixed positive integer. Then, both the $(p,k)$-\textsc{Square Covering} and the $(p,k)$-\textsc{Rectangle Covering} problem have an $\Omega(n \log n)$ lower bound in the algebraic decision tree model.
\end{lemma}
\begin{proof}
	We reduce from 1-dimensional set disjointness: Given a sequence $S=\{r_1,\dots,r_n\}$ of $n$ real numbers, we want to decide whether there is any repeated element in $S$. The following works for both squares and rectangles.

	Given the sequence $S$, we generate the point set $\mathcal{S} = \{(r_i, r_i) \mid 1\leq i \leq n\}\in \mathcal{R}^2$. We compute the $p$ minimal squares that cover $\mathcal{S}$, allowing exactly $k=n-p-1$ outliers, which means that the union of the $p$ squares must cover $p+1$ points. Thus, the covering squares  degenerate to points (i.e., squares of side length zero) if and only if there is a repeated element in the sequence. Otherwise, by the pigeon hole principle, one of the covering squares must cover at least two points and hence, has positive area.
\end{proof}

	Similar bounds for slightly different problems were given by Chan~\cite{c-garot-99} ($p=1$) and by Segal~\cite{s-lbcp-02}
	($p=2,k=0$, arbitrary orientation).

\section{NP-Hardness Results}\label{sec:nphard}

	In this section, we show that both the $(p,k)$-\textsc{Square Covering} and the $(p,k)$-\textsc{Rectangle Covering} problems are NP-hard for any fixed $k$ when $p$ is part of the input. In the following, we focus on the decision version of the two problems for $k=0$: Given $n$ points in the plane and an integer $p > 0$, decide whether or not there exist $p$ axis--aligned unit squares or $p$ axis--aligned rectangles of area at most one that together cover all points. We reduce from planar 3-SAT. Note that we are not dealing explicitly with outliers. However, the reduction can be adapted by placing $k$ points at a sufficiently large distance from the other points as not to be included in the covering. Furthermore, note that our reductions work for all possible cases where the squares or rectangles may (not) overlap or need (not) be congruent. The optimal solutions may be different, however, depending on the underlying case.

\subsection{Covering Points with Squares}\label{sec:hard}
	In this section we study the complexity of the \sq\ problem: cover $n-k$ points in the plane with $p$ axis-aligned squares while minimizing  the area of the largest square.

	NP-hardness of the $p$-center problem (i.e., covering with congruent squares which are allowed to overlap) has been shown previously by Fowler et al.~\cite{fpt-opcpa-81}, and by Meggiddo and Supowit~\cite{ms-cscgl-84}. Here we show NP hardness for the case of covering by congruent squares that must not overlap (except at their boundaries).

	We reduce from planar 3-SAT: given a 3-CNF formula $F$ with variables $x_1, \ldots, x_n$ and clauses $c_1, \ldots, c_m$, let $G(F)$ be the {\em graph} of $F$, defined as:
\begin{itemize}
	\item 	$V=\{x_i \mid 1\leq  i \leq n \} \cup \{c_j \mid 1\leq  j \leq m \}$
	\item 	$E=\{(x_i,c_j)\mid x_i \in c_j$ or $\overline{x_i} \in c_j\}$
\end{itemize}

	If $G(F)$ is a planar graph, then $F$ is called a planar $3$-CNF formula. It is NP-hard to decide whether a given planar $3$-CNF formula is satisfiable or not~\cite{l-pftu-82}.

\subsubsection{Reduction}

\begin{figure}
\begin{center}
\includegraphics[width = 0.3\textwidth]{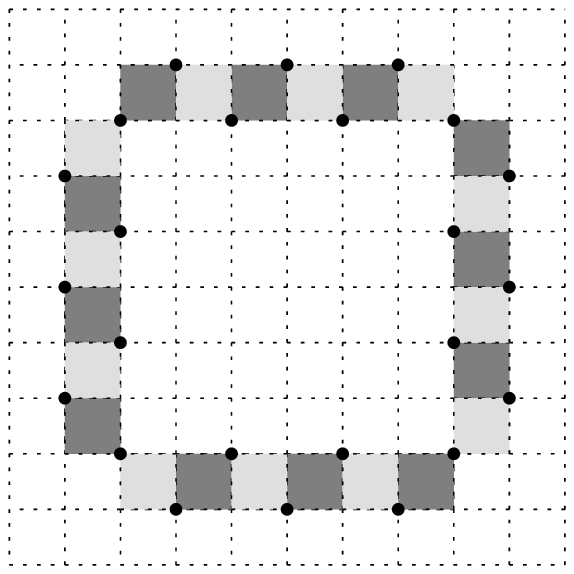}
\includegraphics[width = 0.55\textwidth]{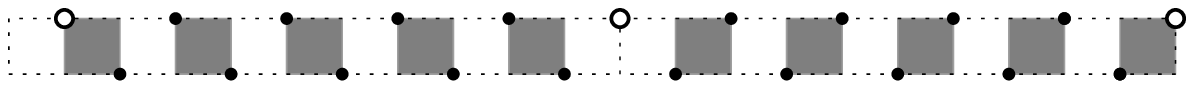}
\caption{Left: Variable gadget consisting of $4N$ points that can be covered in two different ways with $2N$ unit squares (either light or dark grey). Right: Clause gadget with $4M + 1$ points (including three link points - depicted as hollow circles). $2M$ boxes are necessary and sufficient to cover all points except any one of the link points.}
\label{fig_vargadget}
\end{center}
\end{figure}

	Given a planar 3-SAT instance, we construct a \textsc{$(p,0)$-Square Covering} instance on a grid such that the 3-SAT instance is satisfiable if and only if all points can be covered by $p$ unit squares. The reduction is as follows, with all points lying on a grid, such that the $L_{\infty}$ distance between two points in the same grid cell is one unit.

\begin{itemize}
	\item 	For each variable $x_i$, we create a gadget of $4N$ points arranged in a ring-like fashion (where $N$ is a sufficiently large constant). By construction, there are only two different ways of covering all generated points with $2N$ unit squares (see Figure~\ref{fig_vargadget}, left). We associate each of the coverings to an assignment either of \true or \false to the literal, and define the {\em \true region} as the union of squares in the \true assignment, and the {\em \false region} as the union of squares in the \false assignment.

	\item 	For each clause $c_j$, we generate $4M+1$ points in a linear fashion, where $M$ is another large constant. There are three special {\em link points} in the gadget: the rightmost, leftmost and middle points of the linear segment, depicted as hollow circles in Figure~\ref{fig_vargadget}, right.
\end{itemize}

	The main property of the clause gadget is the following:

\begin{lemma}\label{lem_clause}
	To cover all points of a clause gadget except for any one of the three link points, $2M$ unit squares are sufficient and necessary.
\end{lemma}
\begin{proof}
	Figure~\ref{fig_vargadget}, right, shows a covering of all points (except for the middle link point) with $2M$ squares. By shifting the $M$ rightmost (or leftmost) squares to the center, we can cover the middle link, but at the same time we uncover the right (or left) link point; therefore the upper bound holds.

	Consider any covering of all non-link points, which forms two sequences of equal length to the left and right of the middle link point, that are more than unit distance apart. We need at least $\left\lceil  (2M-1)/2 \right\rceil=M$ squares to cover each point sequence, thus the lower bound also holds.
\end{proof}

\begin{figure}
\centering
\includegraphics[width=0.8\textwidth]{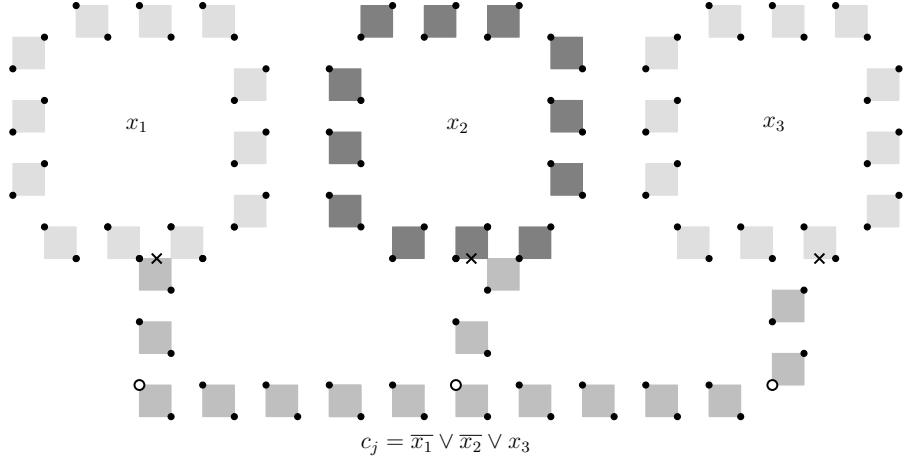}
\caption{Connection between a clause gadget and its corresponding variable gadgets (switches depicted as crosses and links as hollow circles).
In the clause $c_j$, $x_1$ and $x_2$ are negated --- their switch lies in the \false region, whereas $x_3$ is non negated in $c_j$ --- the switch lies in the \true region. The assignment of $x_1, x_3$ -- \true (light grey) and $x_2$ -- \false (dark grey), which satisfies the clause $c_j$, leads to a covering of all connecting points and the clause gadget. }
\label{fig_connect}
\end{figure}

	We connect each clause gadget with its three corresponding variable gadgets as follows (see Figure~\ref{fig_connect}): from each link point of a clause $c_j$ we add a sequence of {\em connecting points} leading to one variable. Let $e_{1,j}$ ($e_{2,j},e_{3,j}$, resp.) be the total number of points added to connect clause gadget $c_j$ with the variable gadgets $x_{1}$ ($x_{2},x_{3}$, resp.). We set $e_{i,j}$ to be odd, which can always be done by making the underlying grid sufficiently fine.

	For each connection between clause gadget $c_j$ and the variable gadgets $x_{1}, x_{2}$ and $x_{3}$, we add three additional points called {\em switches} $s_{1,j}, s_{2,j}$ and $s_{3,j}$. We put the switches between two points of the outer boundary of the variable gadget, either in its
	FALSE or TRUE region, depending on whether the associated literal is negated or not. This way the switch is already covered by a square of the variable gadget if and only if the corresponding variable assignment makes the literal TRUE. We say that the switch is {\em on} if it is covered by a square of the variable gadget, and {\em off} otherwise. Figure~\ref{fig_connect} shows how to connect the clause gadget $c_j$ with the three variable gadgets when the specific assignment of truth values is \true for $x_1, x_3$ and \false for $x_2$.

\begin{lemma}\label{lem_switch}
	Any clause gadget $c_j$ and its connecting points can be covered with $2M+\sum_{i=1}^3\left\lceil e_{i,j}/2\right\rceil$ unit squares if and only if at least one switch is {\em on}.
\end{lemma}
\begin{proof}
	Consider the covering of the connecting points when the corresponding switch is {\em off}, i.e., it is not covered by a square of the associated variable gadget. In this case, the first square of the connection must cover both the switch and the first connecting point. The following squares cover the second and third connecting points, etc. Since the number of connecting points is odd, the last square covers the last two connecting points.

	If the switch $s_{i,j}$ is {\em on}, i.e., it lies in the covering of the variable, then the first square of the connection can be moved to cover the first and second connecting points, the second square covers the third and fourth connecting points, and the last square covers the last connecting point and the $i$th link point of the clause gadget $c_j$.

	Clearly $\sum_{i=1}^3\left\lceil e_{i,j}/2\right\rceil$ squares are necessary to cover all connecting points, thus the remainder of this lemma follows directly from Lemma \ref{lem_clause}.
\end{proof}

	Since $G(F)$ is planar, there exists an embedding of our construction so that no two connections overlap. Furthermore, since $N$ is large (in particular larger than the degree of $G(F)$), we can place switches far away from each other (i.e., more than two units away from each other) so that the associated coverings are independent. Using the lemma above we derive the following lemma:

\begin{lemma}\label{theo:npsq}
	A planar $3$-SAT formula is satisfiable if and only if the associated point covering problem instance can be covered with $2nN+2mM+E$ unit squares, where $E=\sum_{j=1}^m\sum_{i=1}^3\left\lceil e_{i,j}/2\right\rceil$.
\end{lemma}
\begin{proof}
	($\Leftarrow$): Consider any covering of the points. Using Lemma~\ref{lem_switch} and the pigeon hole principle, $2nN$ unit squares are needed to cover all variable gadgets and at least $2mM+E$ unit squares are necessary to cover all clause gadgets (including the connecting points and switches). Thus, each variable must be covered with exactly $2N$ squares and each clause must use exactly $2M+\sum_{i=1}^3\left\lceil e_{i,j}/2\right\rceil$ squares.

	In particular, the covering for the variables is fixed; hence any covering gives a valid variable assignment. By Lemma~\ref{lem_switch} we get that at least one switch must be {\em on} for each clause. This corresponds to each clause $c_j$ being satisfied at least once; thus the 3-SAT instance as a whole is satisfied.

	($\Rightarrow$): Given a variable assignment, we generate the corresponding covering. By construction, each clause $c_j$ must have at least one switch {\em on}, therefore the gadget of $c_j$ (and its connecting points) can be covered using $2M+\sum_{i=1}^3\left\lceil e_{i,j}/2\right\rceil$ squares.
\end{proof}

	The following lemma on hardness of approximation follows from our construction above:

\begin{lemma}\label{cor:insat}
	If the $3$-SAT formula is not satisfiable, any covering with $2nN+2mM+E$ squares has at least one square with area at least $9/4$.
\end{lemma}
\begin{proof}
	By construction, all points have integer coordinates (semi integer if the point is a switch). That is, all points can be written as $p=(u+k/2,v+k/2)$, where $u,v \in \mathbb{N}$ and $k \in \{0,1\}$. Assume that there exists a covering which has a largest square with area strictly smaller than $9/4$ (i.e., the largest square has side length smaller than $3/2$). Given any square covering of the construction, we shrink each square until it has two points on opposite sides of the boundary, without uncovering any points. By shrinking the squares, we set the side length of each square to the difference in either $x$- or $y$-coordinates of some two points of the construction. Since by Lemma~\ref{theo:npsq} it is not possible to find a covering with unit squares, the next possible side length is $3/2$.

\end{proof}

	We conclude this section with the following theorem:

\begin{theorem}
	Given $n$ points in the plane, let $p \in \mathbb{N}$ be part of the input and let $k$ be any fixed integer with $n-k \in \Omega(n)$.
	Then, the $(p,k)$-\textsc{Square Covering} problem is NP-hard. Moreover, it is NP-hard to find an approximate solution within ratio $2.25$.
\end{theorem}


\subsection{Covering Points with Rectangles}
	In this section we show NP-hardness for the \rec\ problem. Note that by making an affine transformation of the previous reduction for squares, we can easily obtain hardness for coverings with rectangles of any fixed ratio. However, the reduction does not work for arbitrary rectangles, since in this case we can cover each variable gadget with eight horizontal and vertical segments of zero area (i.e.,  arbitrarily thin rectangles). By doing so, all switches will be {\em on}, regardless of the variable assignment, and the reduction fails. Hence, we need a different reduction for the \rec\ problem. Again, we reduce from planar 3-SAT, and focus on the decision version of the problem for $k=0$. We call an axis--aligned rectangle a {\em unit rectangle} if its area is at most one, and $p$ unit rectangles form a {\em unit covering} if they together cover all points.

\subsubsection{Staircase sequences}
	For our reduction, we need the notion of {\em staircase} sequences:

\begin{defi}
	A sequence $S=(p_1, \ldots p_{2N})$ of $2N$ points in the plane is a {\em staircase sequence} if and only if it satisfies the following properties:
\begin{itemize}
	\item For any integer $0 \leq i < N$, two consecutive points $p_{2i}$ and $p_{2i+1}$ of the sequence have the same $x$-coordinate and two consecutive points $p_{2i-1}$ and $p_{2i}$ have the same $y$-coordinate (we assume the sequence is closed and set $p_{2N} = p_0$).
	\item No unit rectangle covers any two non-consecutive points of $S$.
\end{itemize}
	We call $n$ staircase sequences $S_1, \ldots S_n$ {\em mutually independent} if no unit rectangle contains points of more than one sequence.
\end{defi}

	We will consider a covering of points that can be decomposed into mutually independent staircase sequences. By definition, no unit rectangle can include points of two independent sequences, thus the coverings of each sequence can be considered independently.

	Consider any unit covering of a single staircase sequence of $2N$ points with $N$ rectangles. If we cover successive points by horizontal or vertical segments, we obtain a covering with largest area zero. We call the covering of a staircase sequence \emph{vertical}, if the sequence is covered by $N$ rectangles such that each rectangle contains two points with the same $x$-coordinate. Similarly, we call the covering of a staircase sequence \emph{horizontal}, if the points inside one rectangle have the same $y$-coordinate, see Figure~\ref{fig:staircase}.

\begin{figure}
\centering
\includegraphics[width = 0.35\textwidth]{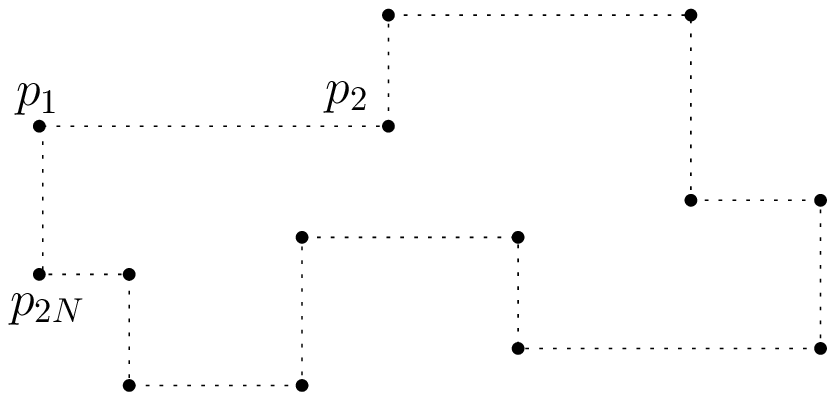}
\caption{Staircase sequence of $2N$ points. Selecting either the horizontal or vertical segments are the only ways of covering the sequence with $N$ unit rectangles.}
\label{fig:staircase}
\end{figure}

\begin{lemma}
	Any unit covering of a staircase sequence of $2N$ points with $N$ rectangles must either be a vertical or a horizontal covering.
\end{lemma}
\begin{proof}
	By the definition of staircase sequence no unit rectangle can cover three points. Therefore, each covering rectangle must contain exactly two consecutive points. Since the rectangles must be disjoint, either all rectangles cover two points with the same $x$-coordinate or all rectangles cover points with the same $y$-coordinate.
\end{proof}

	$N$ unit rectangles are both necessary and sufficient to cover a staircase sequence of $2N$ points, therefore we have:

\begin{cor}\label{cor_assig}
	Any unit covering of $n$ mutually independent staircase sequences, each with $2N$ points, that uses $n N$ rectangles must have either a vertical or a horizontal covering for each sequence.
\end{cor}
\subsubsection{Reduction}
	We construct $n$ mutually independent staircase sequences of $2N$ points each, where $n$ is the number of variables in the associated 3-SAT instance. Any unit covering of the points with $n N$ rectangles gives a variable assignment as follows: variable $x_i$ is set to \true if the $i$th staircase sequence has a horizontal covering, and \false otherwise. Similar to the square case, we add one more point for each clause. This point can only be covered by a unit rectangle if the corresponding variable assignment satisfies the clause.

\begin{figure}
\centering
\includegraphics[width=0.8\textwidth]{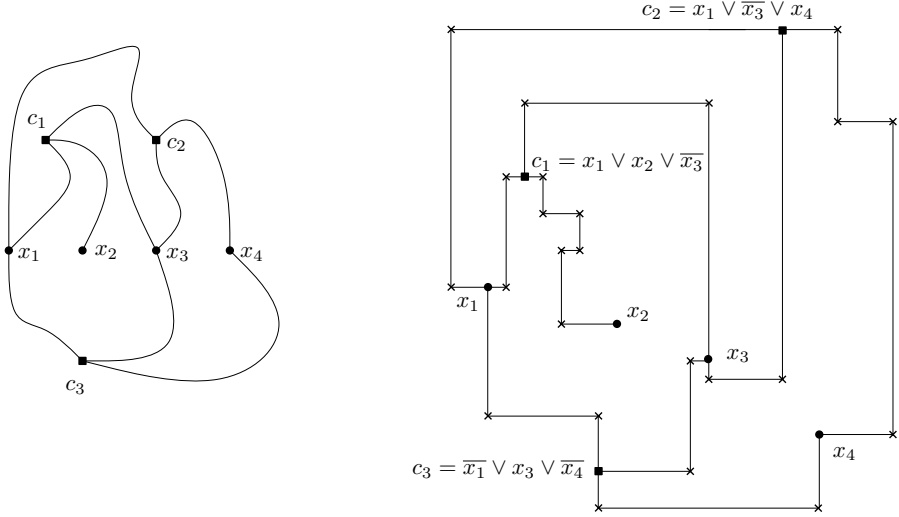}
\caption{As $G(F)$ is planar, we can transform any plane embedding of $G(F)$ into a rectilinear drawing such that each rectilinear tree has $N-2$ bends and non-adjacent bends do not have the same $x$- or $y$-coordinates.}
\label{fig_transformation}
\end{figure}

	Recall that $G(F)$ is planar, thus there exists a planar embedding of $G(F)$ such that all edges can be drawn as rectilinear arcs in the unit grid. For simplicity, we first consider the case in which there is at least one negated and one non-negated literal in each clause (we will show how to deal with the other types of clauses later). We call the union of all rectilinear arcs that connect some variable $x_i$ to the $1 \leq k \leq m$ clauses containing $x_i$ a \emph{rectilinear tree}. That is, we consider the variable node as the root, and the $k$ clause nodes as the leaves, and we choose an embedding for each tree such that the root and each internal node has degree exactly three and the whole tree has exactly $N - (k+1)$ bends. As $G(F)$ is planar, and we can choose $N$ sufficiently large, this is always possible. Consider now the rectilinear arc connecting variable $x_i$ with clause $c_j$. We modify the embedding such that the component of a tree incident to clause $c_j$ is vertical if the literal $\ell_i$ is negated in $c_j$, and horizontal otherwise, which is also always possible. We then further perturb the embedding such that no two non-successive bends of any rectilinear arcs have the same $x$- or $y$-coordinate. Finally, to avoid overlap when thickening the trees (as explained in the next paragraph), we scale the embedding by a factor $2(n+1)$, see Figure~\ref{fig_transformation} for an illustration.

\begin{figure}
\centering
\includegraphics[width=0.6\textwidth]{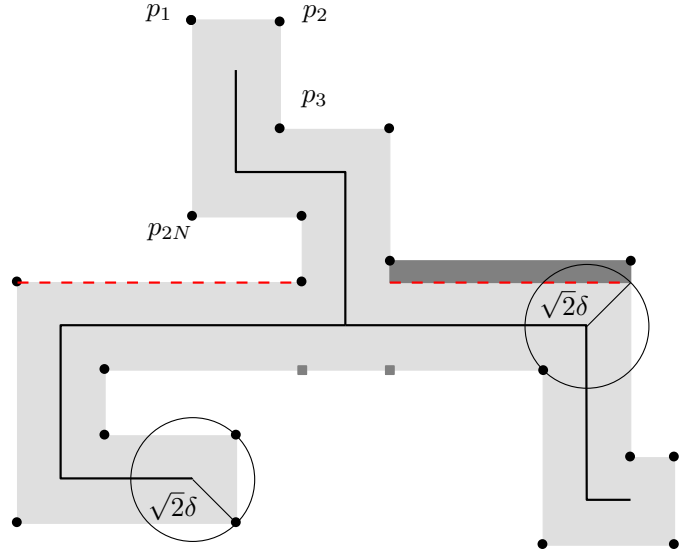}
\caption{Thickening of rectilinear trees results in a staircase sequence. For each endpoint or bend of the tree two new points at distance $\sqrt{2}\delta$ are generated. When an edge is split (dashed segments) we add unit squares until no non-adjacent edges of the sequence have the same $x$- or $y$-coordinate. We ignore the points that lie on the boundary of another thickened path (grey squares).}
\label{fig_thick}
\end{figure}

	We now replace each rectilinear tree containing $N-(k+1)$ bends and $k+1$ endpoints (one of them a variable, the $k$ others clause nodes) by a staircase sequence of $2N$ points as follows (see Figure~\ref{fig_thick}). We arbitrarily assign to each of the $n$ rectilinear trees in $G(F)$ a unique number $\delta \in \{1,\ldots,n\}$ and replace it by a path that is the Minkowski sum of the tree and a square of side length $2\delta$. Each rectilinear tree becomes a set of thickened paths that form a rectilinear polygon. Note that at any internal node (or the root), one of the vertical or horizontal components will split into two parts. When this happens, we add unit squares to the polygon until no non-consecutive edges of the polygon have the same $x$- or $y$-coordinate, without changing the number of polygon vertices which is always possible. Furthermore, two endpoints of one thickened path will lie on the boundary of one of the other thickened paths. These two points can be ignored. We then walk along the boundary of the generated polygon, and number the vertices in clockwise order; let $S_i=(p_1, \ldots, p_{2N})$ be the sequence of generated vertices.

\begin{lemma}
	The sequences $S_1, \ldots, S_n$ of vertices generated as above form $n$ mutually independent staircase sequences, each of them containing $2N$ points.
\end{lemma}

\begin{proof}
	With the above transformation, we get the following new coordinates for the vertices of a tree. Let $P=(X,Y)$ be a node of the tree before both the scaling and the thickening, with integer coordinates. After the scaling with factor $2(n+1)$ it has the coordinates $P'=((2n+2)X, (2n+2)Y)$. After the thickening with factor $\delta$, the node transforms into a pair of vertices $p_{1,2}$, that lie on a circle $C$ with radius $\sqrt{2}\delta$ centered at $P'$. Depending on whether $P$ is an endpoint (i.e., the root or a leaf) or a bend of the original tree, these two vertices either lie on a quadrant or on a diameter of $C$. As all the numbers involved are integer, we get for each node $P$ of the tree a vertex pair with coordinates $p_{1,2} =((2n+2)X\pm \delta \pm k ,(2n+2)Y\pm \delta \pm k)$. Here, $X$ and $Y$ are integers, $\delta \leq n$ is the thickening factor, $k\in \{0,1\}$ is a factor describing the possible addition of unit squares to avoid having the same coordinates in non-adjacent edges, and $|2\delta+k| < 2n+2$. Therefore, two points can be covered by a unit rectangle if and only if they share one coordinate. This can only happen when both points are adjacent on the generated staircase sequence.
\end{proof}

\begin{figure}
\centering
\includegraphics[width=0.75\textwidth]{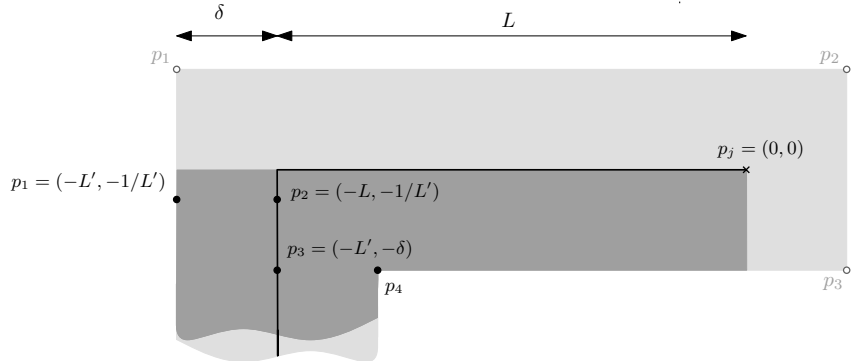}
\caption{Local transformation around clause $c_j$ (corresponding to point $p_j$). Staircase before (light grey) and after (dark grey) moving points to avoid intersection with other staircase sequences.}
\label{fig_trans}
\end{figure}

	By construction, the generated staircase sequences do not intersect, except at the clause variables. To remove these intersections, we modify the sequences locally around each clause node. Consider only a small neighborhood of clause $c_j$, and assume that we have a segment of length $L$ connecting to $c_j$ from the left (see Figure~\ref{fig_trans}). We add a point $p_j$ at the position of node $c_j$ to the staircase sequence.

	Assuming that $p_j=(0,0)$, we define $L'=L+\delta$ (where $\delta$ is the thickness of the path) and move the three points located at $(-L',\delta)$, $(\delta,\delta)$ and $(\delta,-\delta)$ to the new coordinates $p_1=(-L',-1/L')$, $p_2=(-L,-1/L')$ and  $p_3=(-L,-\delta)$. When connecting from below, right, or above, we use appropriately rotated versions of the transformation described above.

	Points $p_1$ and $p_2$ are called the {\em links} between clause $c_j$ and variable $x_i$. The main property of the construction is that we can cover both link points and the point $p_j$ with a single rectangle of area one. It is easy to see that the new coordinates of the three moved points are rational and that the staircase sequences remain mutually independent.

\begin{figure}
\centering
\includegraphics[]{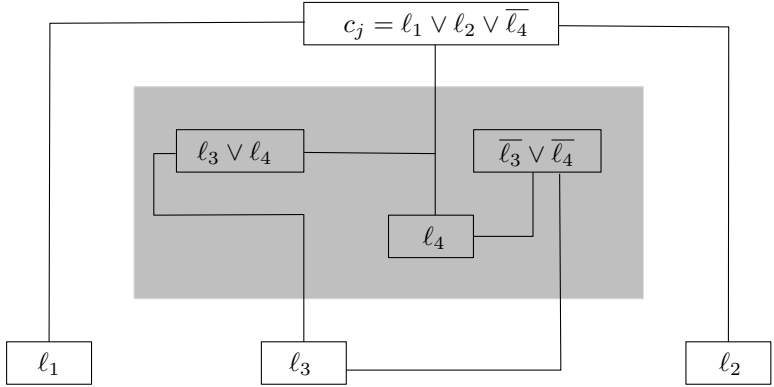}
\caption{Local transformation for clause $c_j= \ell_{1} \vee \ell_{2} \vee \ell_{3}$: using a negation gadget (inside the grey box) we can negate a literal in $c_j$.}
\label{fig_planar}
\end{figure}

	Finally, we need to show how to deal with clauses with all three literals either negated or not. This is important, as we cannot have three horizontal or vertical connections to the same clause node. Let $c_j= \ell_{1} \vee \ell_{2} \vee \ell_{3}$ be such a clause, then we can transform it into the following three clauses: $(\ell_{1} \vee \ell_{2} \vee \overline{\ell_{4}}) \wedge (\ell_{3} \vee \ell_{4})\wedge (\overline{\ell_{3}} \vee \overline{\ell_{4}})$. Here, $\ell_{4}$ is a literal of a new variable, and the two last clauses assure that $\ell_4$ has the opposite truth assignment of $\ell_3$.

	For each such clause, we additionally generate only one variable and two clauses, thus the asymptotical size of the transformation as well as its planarity are not affected (see Figure~\ref{fig_planar}). This transformation needs only constant space, hence can be done independently for each literal. After transforming all such clauses we can proceed as before.

	Let $P$ be the set of $2nN+m$ points of the $n$ staircase sequences generated by the transformation of a 3-SAT formula with $n$ variables and $m$ clauses. We have arrived at the following lemma.
\begin{lemma}
	A planar 3-SAT formula in $n$ variables is satisfiable if and only if the set $P$ of $2nN + m$ points generated as above can be covered with $nN$ unit rectangles.
\end{lemma}
\begin{proof}
	($\Leftarrow$) Given a unit covering of $P$, we generate a variable assignment as follows: each variable is set to \true if its associated staircase sequence has a horizontal covering, \false otherwise. As any unit covering of $P$ is a unit covering of the $n$ mutually independent staircase sequences, this assignment is valid by Corollary~\ref{cor_assig}.

	We now show that this variable assignment satisfies all clauses; by construction, any rectangle that covers at least four points has area larger than one, thus no such rectangle can be in a unit covering. Since there are $2nN+m$ points in the construction and we want to cover them with $nN$ rectangles, there must be exactly $m$ rectangles, each covering three points. No three points from a variable gadget can be covered with a unit rectangle, thus each of the $m$ rectangles must cover two variable points and the point $p_j$ corresponding to clause $c_j$.

	By construction of the clause node $p_j$, such a covering is only possible if $p_j$ and any two links are covered by the same rectangle. Let $x_i$ be the variable with two links that are covered together with $p_j$ by one unit rectangle. If the literal of $x_i$ is not negated in $c_j$, the links share the $y$-coordinate. Since both links are covered by the same rectangle, the gadget of $x_i$ must be horizontally covered, which corresponds to setting variable $x_i$ to \true in our variable assignment. Since $x_i$ is set to \true and literal $\ell_i$ is not negated, clause $p_j$ is satisfied. The case with negated $\ell_i$ is analogous.

	($\Rightarrow$): Given a variable assignment, we generate a corresponding covering for the gadget variables. Each clause $c_j$ is satisfied at least once, thus we can cover point $p_j$ together with the link points of the variable that satisfies $c_j$ with one unit rectangle.
\end{proof}

	For the $(p,k)$-\textsc{Rectangle Covering} problem we can give the following inapproximability result:

\begin{lemma}
	If the 3-SAT formula is not satisfiable, any covering of the $n$ staircase sequences with $nN$ rectangles has at least one arbitrarily large rectangle.
\end{lemma}
\begin{proof}
	We scale the transformation by an arbitrarily large, constant factor $M$ before the local transformation in the neighborhood of the variables is done. If the 3-SAT formula is satisfiable, a unit covering is possible. However, consider any covering of a non-satisfiable 3-SAT instance: since the thick paths become arbitrarily thick, horizontal and vertical coverings are forced, and thus each covering still gives a valid variable assignment.

	We must enlarge the rectangles such that they cover all clause points $p_j$. Since the instance is non-satisfiable, for any variable assignment there exists a clause $c_j = \ell_1 \vee \ell_2 \vee \ell_3 $ with vertically covered variables if the literal is not negated, and horizontally covered variables otherwise. The minimum area rectangle that includes $p_j$ and two points sharing a $y$-coordinate (if the literal is not negated) includes the points $p_2$ and $p_3$, and it has area $M^2 L' \delta = M^2 \delta L + M^2 \delta^2$, which is arbitrarily large.
\end{proof}

\begin{theorem}
	Given $n$ points in the plane, let $p\in \mathbb{N}$ be part of the input and $k$ be any fixed integer with $n-k \in \Omega(n)$.
	Then, the $(p,k)$-\textsc{Rectangle Covering} problem is NP-hard.
	Moreover, the $(p,k)$-\textsc{Rectangle Covering} problem admits no constant-factor polynomial time approximation algorithm.

\end{theorem}

\section{Exact Algorithms for $p\leq 3$}\label{sec:smallp}\label{sec:algs}


	In this section, we present algorithms to efficiently compute the solution for the \textsc{$(p,k)$-Box Covering} problem for small values of $p$. For simplicity, we assume throughout the following sections that no two points have the same $x$- or $y$-coordinate, and we assume furthermore in the description of our algorithms that we want to cover \emph{exactly} $n-k$ points. An adaptation to cover \emph{at least} $n-k$ points is straightforward.  Note that for $p\in \{2,3\}$, we can always find an axis parallel line that separates one box from the others. We exploit this property for a divide-and-conquer type of approach.

\subsection{Covering Points with Squares}\label{sec:squarecov}
	We first want to cover $n-k$ points of $P$ with $p$ squares. With a simple observation, we can improve an existing algorithm for computing the optimal solution of the \textsc{$(1,k)$-Square Covering} problem, which will function as our base case. Using certain monotonicity properties, we can apply binary search.
\subsubsection{$(1,k)$\textsc{-Square Covering}}\label{sec:squares}

	Previously, an $O(n\log n)$ expected time algorithm for the $(1,k)$-\textsc{Square Covering} problem was presented by Chan~\cite{c-garot-99}. We make use of Chan's algorithm as a subroutine of our algorithms.

	A point $p \in P$ is called {\em $(k+1)$-extreme} if either its $x$- or $y$-coordinate is among the $k+1$ smallest or largest in $P$. Let $E(P)$ be the set of all $(k+1)$-extreme points of $P$.

\begin{lemma}\label{lem:preproc}
  	For a given set $P$ of $n$ points in the plane, we can compute the set $E(P)$ of all $(k+1)$-extreme points of $P$ in $O(n)$ time.
\end{lemma}
	We can use the standard selection algorithm~\cite{clrs-ia-01} to select the point $p_L$ of $P$ with $(k+1)$-st smallest $x$-coordinate in linear time. We then go through $P$ again to find all points with $x$-coordinate smaller than $p_L$. Finding the points $p_R, p_T, p_B$ and computing the rest of $E(P)$ is symmetric.

	The following lemma shows that the left side of the optimal solution of the \textsc{$(1,k)$-Square Covering} problem lies on or to the left of the vertical line through $p_L$, and that the right side lies on or to the right of the vertical line through $p_R$. Similarly, the top side of the optimal solution lies on or above the horizontal line through $p_T$, and the bottom side lies on or below the horizontal line through $p_B$.

\begin{lemma}\label{lem:key1} 
  	The optimal square $\optSQ$ that solves the \textsc{$(1,k)$-Square
    	Covering} problem is determined by the points of $E(P)$ only.
\end{lemma}
\begin{proof}
	The covering square is convex, hence all outliers must come from outside the optimal square. As we want to minimize the area, there exists an optimal square $B^*$ such that at least three edges of $B^*$ each contain one point of $P$. If one edge, say the top edge, is determined by a point $p\in P\setminus E(P)$, it means that there are at least $k+1$ outliers above $B^*$, which is not allowed.
\end{proof}
        Using this lemma, we obtain an improved running time as follows:
\begin{theorem}\label{thm:1k}
  	Given a set $P$ of $n$ points in the plane, the \textsc{$(1,k)$-Square Covering} problem can be solved in $O(n+k\log k)$ expected time using $O(n)$ space.
\end{theorem}
\begin{proof}
	We first compute the set of extreme points $E(P)$ in linear time and then run Chan's algorithm on the set $E(P)$. The time bound follows directly, since $|E(P)|  \leq 4k+4$.
\end{proof}

\subsubsection{$(2,k)$\textsc{-Square Covering}}
\label{sec:2ksquares}

	The following observation is crucial to solve the  \textsc{$(2,k)$-Square Covering} problem, where we look for two disjoint squares that cover $n-k$ points.
\begin{obs}\label{obs:key2}
  	For any two disjoint axis-aligned squares in the plane, there exists an axis-parallel line  $\ell$ that separates them.
\end{obs}
	This observation implies that there is always an axis-parallel
	line $\ell$ that separates the two optimal squares $(B^*_1, B^*_2)$ of the solution of a
	\textsc{$(2,k)$-Square Covering} problem. Let $\ell^+$ be the
	halfplane defined by $\ell$ that contains $B^*_1$.  Let $P^+$ be the
	set of points of $P$ that lie in $\ell^+$ (including points on $\ell$),
	and let $k^+$ be the
	number of outliers admitted by the solution of the \textsc{$(2,k)$-Square
  	Covering} problem that lie in $\ell^+$. Then
	there is always an optimal solution of the  \textsc{$(1,k^+)$-Square
  	Covering} problem for $P^+$ with size smaller than or
	equal to that of $B^*_1$.
	The same argument also holds for the other halfplane $\ell^-$, where we have $B^*_2$, and
	$k^-=k-k^+$.  Thus, the pair of optimal solutions of $B_1^*$ of the $(1,k^+)$-\textsc{Square Covering} problem and $B_2^*$ of the \textsc{$(1,k^-)$-Square Covering} problem is  an optimal solution of the original
	\textsc{$(2,k)$-Square Covering} problem.

\begin{figure}[t]
  \centering
  \includegraphics[width=0.70\textwidth]{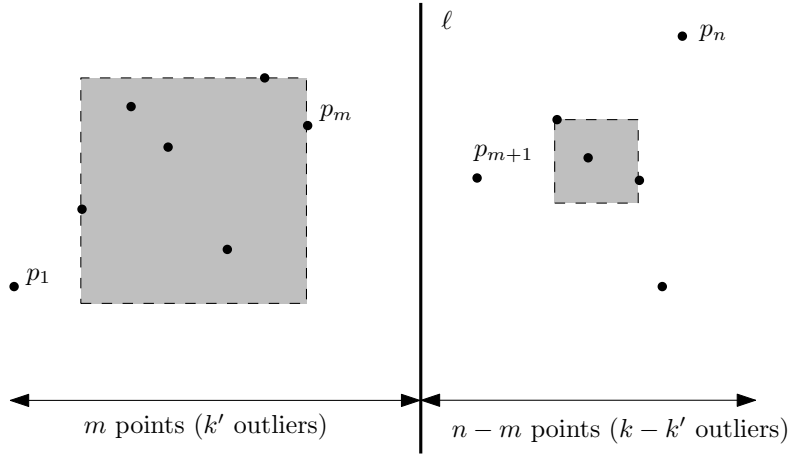}
  \caption{For given $k'$ and $m$, the optimal $m^*$ lies on the side
    with the larger square. }
  \label{fig:2squarebin}
\end{figure}

\begin{lemma}
	There exists an axis-parallel line $\ell$ and a positive integer
	$k'\leq k$ such that an optimal solution of the $(2,k)$-\textsc{Square Covering}
	problem for $P$ consists of the optimal solution of the $(1,k')$-\textsc{Square
	Covering} problem for $P^+$ and the $(1,k-k')$-\textsc{Square Covering} problem for $P^-$.
\end{lemma}

	We assume w.l.o.g. that $\ell$ is vertical, and we associate  $\ell$
	with $m$, the number of points that lie to the left of (or on) $\ell$.
	Let $p_1, p_2,\ldots,p_n$ be the list of points in $P$ sorted
	by $x$-coordinate. Then $\ell$
	partitions the points of $P$ into two subsets, a left point set,
	$P_L(m)=\{p_1, \ldots, p_{m}\}$ and a right point set,
	$P_R(m)=\{p_{m+1}, \ldots, p_n\}$, see Figure~\ref{fig:2squarebin}.
	The optimal left square is
	a solution of the \textsc{$(1,k')$-Square Covering} problem for $P_L(m)$
	for $0 \leq k'\leq k$, and the optimal right
	square is a solution of the \textsc{$(1,k-k')$-Square Covering}
	problem for $P_R(m)$.

	We can efficiently compute the optimal solutions for $P_L(m)$ and $P_R(m)$ in each halfplane of a vertical line $\ell$ using the above \textsc{$(1,k)$-Square Covering} algorithm. However, as we have to consider many partitioning lines, it is important to find an efficient way to compute the $(k+1)$-extreme points for each $P_L(m)$ and $P_R(m)$ corresponding to a particular line $\ell$. For this we use Chazelle's segment dragging query algorithm~\cite{c-asdii-88}.

\begin{lemma}[\cite{c-asdii-88}]\label{lem:segdrag}
  	Given a set $P$ of $n$ points in the plane, we can preprocess it in $O(n\log n)$ time and $O(n)$ space such that, for any axis--aligned orthogonal range query $Q$, we can find the point p of P that has the highest $y$-coordinate of all points inside the query range $Q$ in $O(\log n)$ time.
\end{lemma}

	We repeatedly apply Lemma~\ref{lem:segdrag} as follows: We start to query with a rectangle $Q$ that has upper boundary at $+\infty$ to find the topmost point. We then set the upper boundary of the rectangle to the $y$-coordinate of the topmost point and query again with the new rectangle. Doing this $k+1$ times gives the $k+1$ points with highest $y$-coordinate in any halfplane.
	We rotate the set $P$ to find all other elements of $E(P)$ in the according half plane, and we get the following time bound.
\begin{cor}\label{cor:sdq}
  	After $O(n \log n)$ preprocessing time, we can compute the sets
  	$E(P_L(m))$ and $E(P_R(m))$ in $O(k\log n)$ time for any given $m$.
\end{cor}

	Before presenting our algorithm we need the following lemma:
\begin{lemma}\label{lem:increasing_function}
  	For a fixed $k'$, the area of the solution of the \textsc{$(1,k')$-Square Covering} problem for $P_L(m)$ is an increasing function of $m$.
\end{lemma}
\begin{proof}
  	Consider the set $P_L(m+1)$ and the optimal square $B_1^*$ of the
  	\textsc{$(1,k')$-Square Covering} problem for $P_L(m+1)$.  Clearly,
  	$P_L(m+1)$ is a superset of $P_L(m)$, as it contains one more point
  	$p_{m+1}$. Since $k'$ is fixed, the square $B_1^*$ has $k'$ outliers in
  	$P_L(m+1)$. If the interior of $B_1^*$ intersects the vertical line
  	$\ell$ through $p_m$, we translate $B_1^*$ horizontally to the left
  	until it stops intersecting $\ell$. Let $B$ be the translated
  	copy of $B_1^*$, then $B$ lies in the left halfplane of $\ell$ and
  	there are at most $k'$ outliers admitted by $B$ among the points in
  	$P_L(m)$. Therefore we can shrink or translate $B$ and get a square
  	inside the left halfplane of $\ell$ that has exactly $k'$ outliers
  	and a size at most that of $B_1^*$. Thus, the optimal square for $P_L(m)$ has a size smaller or equal to that of $B_1^*$.
\end{proof}

	Lemma~\ref{lem:increasing_function} immediately implies the following corollary.
\begin{cor}\label{cor:monotone}
   	Let $(B_1, B_2)$ be the solution of the \textsc{$(2,k)$-Square Covering} problem with separating line $\ell$ with index $m$. 
   	Then, the index $m^*$ of the optimal separating line $\ell^*$ is at most $m$ if the left square $B_1$ is larger than the right square $B_2$; otherwise it holds that $m^* \geq m$.
\end{cor}

	To solve the \textsc{$(2,k)$-Square Covering} problem, we start with the vertical line $\ell$ at the median
	of the $x$-coordinates of all points in $P$. For a given $m$,
	we first compute the sets $E(P_L(m))$ and $E(P_R(m))$. Then we
	use these sets in the call to the
	\textsc{$(1,k')$-Square Covering} problem for $P_L(m)$ and the
	\textsc{$(1,k-k')$-Square Covering} problem for $P_R(m)$, respectively,
	and solve the subproblems independently.
	The solutions of these subsets give the first candidate for the solution
	of the \textsc{$(2,k)$-Square Covering} problem, and we now compare
	the areas of the two obtained squares. According to Corollary~\ref{cor:monotone}, we can discard one of the halfplanes created by $\ell$ (see Figure~\ref{fig:2squarebin}), hence, we can use binary search to find the optimal index $m^*$ for the given $k'$. As the value of $k'$ that leads to the overall optimal solution is unknown, we need to do this for every  possible $k'$. Finally, we also need to examine horizontal separating lines by reversing the roles
	of $x$- and $y$-coordinates.

\begin{theorem}\label{thm:2ksolution}
  	For a set $P$ of $n$ points in the plane, we can solve the
  	\textsc{$(2,k)$-Square Covering} problem in
  	$O(n \log n+k^2\log^2 k)$ expected time using $O(n)$ space.
\end{theorem}
\begin{proof}
	After $O(n\log n)$ preprocessing time, we have $O(k \log n)$ different queries, each of which takes $O(k \log n)$ time, which gives a total running time of $O(n \log n + k^2 \log n)$. We can show that this is equal to $O(n \log n + k^2 \log^2 k)$ by distinguishing the following two cases:

	If $k^4 \leq n$, then it holds for the second term that $k^2 \log^2 n \leq \sqrt{n}\log^2 n \in O(n)$, so the second term is asymptotically smaller than the first, and we have $O(n \log n + k^2 \log^2 n) = O(n \log n)$.

	If $n< k^4$, then $\log n < \log k^4 \in O(\log k)$, so the second term is asymptotically bounded by $O( k^2 \log^2 k)$, and altogether we have in this case  $O(n \log n + k^2 \log^2 n) = O(n \log n + k^2 \log^2 k)$.
	Hence, in both cases the asymptotic time bound is $O(n \log n + k^2 \log^2 k)$.
\end{proof}

\subsubsection{$(3,k)$\textsc{-Square Covering}}\label{sec:3ksquares}
	The above solution for the \textsc{$(2,k)$-Square Covering}
	problem suggests a recursive approach for the general
	\textsc{$(p,k)$-Square Covering} case: Find an
	axis-parallel line that separates one square from the others and
	recursively solve the induced subproblems. We can do this for $p = 3$, as
	Observation~\ref{obs:key2} can be generalized as follows.

\begin{obs}\label{obs:key3}
  	For any three pairwise-disjoint, axis-aligned squares in the plane, there always exists an
  	axis-parallel line $\ell$ that separates one square from the others.
\end{obs}

	Again we assume that the separating line $\ell$ is vertical and that the left 	halfplane only contains one square. Since Corollary
	\ref{cor:monotone} can be generalized to \textsc{$(3, k)$-Square Covering}, we solve this case as before: fix the amount of outliers permitted on the left halfplane to $k'$ and iterate $k'$ from $1$ to $k$ to obtain the optimal $k^*$. For each possible $k'$, we recursively solve the two subproblems to the left and right of $\ell$ and use the solutions to obtain the optimal index $m^*$ such that the area of the largest square is minimized. Preprocessing consists of sorting the points of $S$ in both $x$- and $y$-coordinates and computing the segment dragging query structure, which can be done in $O(n \log n)$ time.

	In the left halfplane, we solve the $(1,k')$ subproblem as before; its running time is subsumed by the time needed to solve the $(2,k-k')$ subproblem in the right halfplane. Each \textsc{$(2,k-k')$-Square Covering} subproblem is solved as described above, except that preprocessing in the recursive steps is no longer needed: The segment dragging queries can be performed directly since the preprocessing has been done in the higher level. Also, for the binary search, we can use the sorted list of all points in $P$, which is a superset of $P_R(m)$.

	This algorithm has a total time complexity of $O(n \log n+ k^3\log^3 n) = O(n \log n + k^3 \log^3 k)$ (as before by distinguishing $k^6\leq n$ from $k^6>n$).

\begin{theorem}
  	For a set $P$ of $n$ points in the plane, we can solve the \textsc{$(3,k)$-Square Covering} problem in $O(n \log n+k^3\log^3 k)$ expected time using $O(n)$ space.
\end{theorem}


\subsection{Covering Points with Rectangles}\label{sec:rectangles}
	We now look at the \textsc{$(p,k)$-Rectangle Covering}
	problem, where we want to cover $n-k$ points with pairwise--disjoint rectangles. It is straightforward to extend Lemma~\ref{lem:key1} as well as Observations~\ref{obs:key2} and~\ref{obs:key3} to rectangles, so we can use the same approach to solve the
	\textsc{$(p,k)$-Rectangle Covering} problem as for the
	\textsc{$(p,k)$-Square Covering} problem when $p\leq 3$.

	Chan's algorithm~\cite{c-garot-99}, however, does not apply to the $(1,k)$-\textsc{Rectangle Covering} problem, that means that once we have computed the set of $(k+1)$-extreme points, we need to test all rectangles that cover $n-k$ points. Our approach is an exhaustive search: We store the points of $E(P)$ separately in four sorted lists, the top $k+1$ points in $T(P)$, the bottom $k+1$ points in $B(P)$, and the left and right $k+1$ points in $L(P)$, and $R(P)$, respectively. Note that some points may belong to more than one set.

	We first create a vertical slab by drawing two vertical lines through one point of $L(P)$ and $R(P)$ each. All $k'$ points outside this slab are outliers, which leads to $k-k'$ outliers that are still permitted inside the slab. We now choose two horizontal lines through points in $T(P)$ and $B(P)$ that lie inside the slab, such that the rectangle that is formed by all four lines admits exactly $k$ outliers. It is easy to see that whenever the top line is moved downwards, also the bottom line must move downwards, as we need to maintain the correct number of outliers throughout. Inside each of the $O(k^2)$ vertical slabs, there are at most $k$ horizontal line pairs we need to examine, hence we can find the smallest rectangle covering $n-k$ points in $O(k^3)$ time when the sorted lists of $E(P)$ are given. This preprocessing takes $O(n+k \log k)$ time.
	We get the following theorem:
\begin{theorem}
	Given a set $P$ of $n$ points in the plane, we can solve the
  	\textsc{$(1,k)$-Rectangle Covering} problem in  $O(n+k^3)$ time using $O(n)$ space.
\end{theorem}

	Note that this approach leads to the same running time we would get by simply bootstrapping any other existing rectangle covering algorithm~\cite{aiks-fkpmd-91,sk-ekpsa-98} to the set $E(P)$, which has independently been done in~\cite{abcmmpsw-aoor-09}. Note further that for the case of squares, it is possible to reduce the number of vertical slabs that need to be examined to $O(k)$ only, which would lead to a total running time of $O(n + k^2)$.


	The \textsc{$(p,k)$-Rectangle Covering} problem for $p\in \{2,3\}$ can be solved with the same recursive approach as the according $(p,k)$-\textsc{Square Covering} problem, and by using the $(1,k)$-\textsc{Rectangle Covering} algorithm described above as base case. The running times change as follows.
\begin{theorem}\label{theo:rect1}
  	Given a set $P$ of $n$ points in the plane, we can solve the
 	the \textsc{$(2,k)$-Rectangle Covering} problem in $O(n\log n+k^4\log k)$ time, and the \textsc{$(3,k)$-Rectangle Covering} problem in $O(n\log n+k^5\log^2 k)$ time. In both cases we use $O(n)$ space.
\end{theorem}

\section{Concluding remarks}

\begin{figure}
\begin{centering}
\includegraphics[width = 0.25\textwidth]{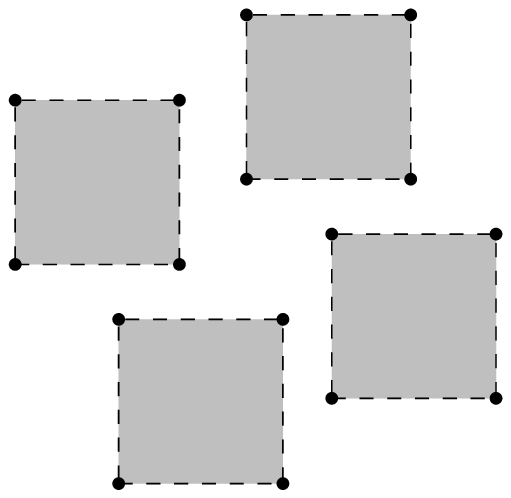}\hspace{1.5cm}
\includegraphics[width = 0.25\textwidth]{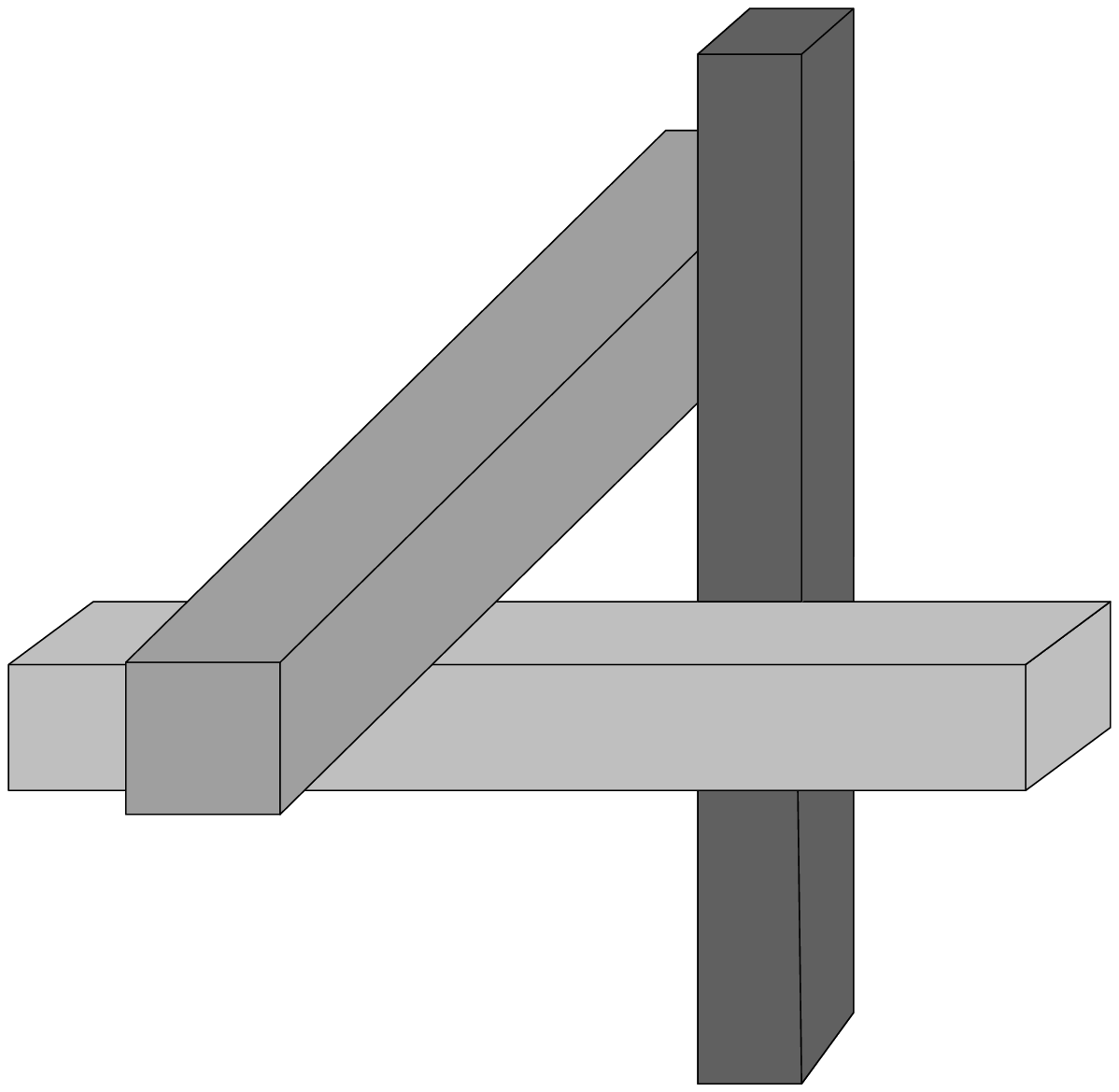}
\caption{Counterexamples. Left: In $\mathbb{R}^2$, no splitting line may exist for $p\geq 4$. Right: In $\mathbb{R}^3$, no splitting hyperplane may exist for $p\geq 3$.}\label{fig:counter}
\end{centering}
\end{figure}

	In this paper we have extended the well examined axis-aligned box covering problem to allow at most $k$ outliers.

	Our algorithms for $p\leq 3$ can be generalized to other functions than minimum area (e.g., minimizing the maximum perimeter of the boxes) as long as this function has some monotonicity property that allows us to solve the subproblems induced by the $p$ boxes independently.

	To solve the $(p,k)$-\textsc{Square Covering} problems we use the randomized technique of Chan~\cite{c-garot-99} as a subroutine, and thus our algorithms are randomized as well. Chan~\cite{c-garot-99} mentioned that his algorithm can be derandomized adding a logarithmic factor. Thus, our algorithms can also be made deterministic, adding an $O(\log k)$ factor to the second term of their running times, see the proof of our Theorem~\ref{thm:2ksolution}.

	We can generalize all algorithms to higher dimensions where the partitioning line becomes a hyperplane. However, there is a simple example (see Figure~\ref{fig:counter}, right), showing that neither the \textsc{$(3,k)$-Square Covering}, nor the \textsc{$(3,k)$-Rectangle Covering} problem admits a partitioning hyperplane for $d>2$, hence our algorithm can only be used for $p=1,2$ in higher dimensions.

	Our algorithms do not directly extend to the case $p\geq 4$, as		Observation~\ref{obs:key2} does not hold for the general case, see Figure~\ref{fig:counter}, left. Although no splitting line may exist, there always exists a quadrant separating a single box from the others. This property again makes it possible to use recursion to solve any \textsc{$(p,k)$-Square Covering} or \textsc{$(p,k)$-Rectangle Covering} problem.

	A natural extension of our idea is to allow either arbitrarily oriented squares and rectangles or to allow them to overlap. Both appears to be difficult within our framework, as we make use of the set of $(k+1)$-extreme points, which is hard to maintain under rotations; also we cannot restrict our attention to only these points when considering overlapping squares or rectangles.\\

\textbf{Acknowledgements:} We thank Otfried Cheong, Joachim Gudmundsson, Stefan Langerman, and Marc Pouget for fruitful discussions on early versions of this paper, and Jean Cardinal for indicating useful references.

\bibliographystyle{abbrv}

\end{document}